\documentclass[11pt]{article}
\usepackage[T1]{fontenc}
\usepackage{lmodern}
\usepackage{amsmath,amssymb,amsthm,mathtools}
\usepackage{amscd}
\usepackage[margin=1in]{geometry}
\usepackage{microtype}
\usepackage[numbers,square]{natbib}
\usepackage[colorlinks=true,linkcolor=blue,citecolor=blue,urlcolor=blue]{hyperref}
\hypersetup{pdftitle={Type classifiers},
 pdfauthor={Martin E. Bidlingmaier}}
\newtheorem{theorem}{Theorem}
\newtheorem{proposition}[theorem]{Proposition}

\theoremstyle{definition}
\newtheorem{definition}[theorem]{Definition}
\theoremstyle{remark}

\newcommand{\C}{\mathcal{C}}
\newcommand{\D}{\mathcal{D}}
\newcommand{\Ty}{\operatorname{Ty}}
\newcommand{\Tm}{\operatorname{Tm}}
\newcommand{\id}{\mathrm{id}}

\newcommand{\Mod}{\mathbf{Cwf}_{\mathbb T}}
\newcommand{\Cxl}{\mathbf{Cxl}_{\mathbb T}}
\newcommand{\cs}[2]{#2\mathbin{/\!/}#1}

\newcommand{\ap}{\mathsf{ap}}
\newcommand{\doi}[1]{\href{https://doi.org/#1}{\nolinkurl{doi:#1}}}
\title{Type classifiers}
\author{Martin E. Bidlingmaier}
\date{}
\begin{document}
\maketitle

\begin{abstract}
We extend dependent type theory by \emph{type classifiers} and construct an interpretation in the category of contextual categories with families.
Type classifiers extend contexts by type variables, complementing the usual context extensions by term variables of specified types.
This allows statements with arbitrary type parameters to be formulated without assuming a universe.

We show that substituting isomorphic types induces isomorphisms between the resulting types even in extensional type theory, giving an analogue of the invariance supplied by the univalence axiom in homotopy type theory.
Omitting function types from the theory, so that all type formers are covariant, extends this naturality to arbitrary, possibly noninvertible maps.
\end{abstract}

\section{Introduction}

Many mathematical assertions involve a free type parameter.
A theorem about \emph{all} groups, for example, must quantify over all possible carrier types.
In dependent type theory, such quantification is usually expressed using universe types, whose elements encode types.
This approach works well in practice, but it has limitations.

Girard's paradox rules out a universe that contains itself.
We can thus either restrict the theorem to small types, those classified by a chosen universe, or introduce an infinite hierarchy $U_0\in U_1\in\cdots\in U_\ell\in\cdots$ and use universe polymorphism to quantify over its levels $\ell$.
But universe polymorphism only shifts the quantification problem, since types quantifying over the whole indexed hierarchy need not belong to any of its universes.
Such types can be interpreted in a larger universe $U_\infty$, but this universe then lies outside the hierarchy over which we quantified.
The quantification still fails to cover all types of the enlarged theory.

A further drawback is that universe closure axioms increase the proof-theoretic strength of the theory~\citep[Section 4]{Palmgren}, often unnecessarily.
For example, when elementary results about groups are formulated using universes, the existence of those universes must be assumed even though they are not needed for the proofs.
The resulting formalization cannot be interpreted directly in small models without those universes, even when the arguments could be carried out there.

We introduce type classifiers for extending contexts by free type parameters, alongside the usual extensions by term parameters.
Their universal property permits substitution by any closed type.
We thus obtain a context classifying all types, rather than a type classifying all types.
A theorem about all groups, for example, is represented by a type in a context containing a free type variable for the carrier type and parameters for the group structure.
It can be specialized by substitution to any particular group.
The theorem is formulated in this extended context, rather than as a dependent product over a universe in the empty context.

The notion of a type classifier is analogous to that of a relative object classifier in topos theory.
Such a classifier exists over every Grothendieck topos, with the universal property that geometric morphisms into it over the base correspond to objects of the domain topos.
In particular, sections correspond to objects of the base topos.
This similarly permits reasoning about arbitrary objects without requiring an internal universe object classifying them.

The main contribution of this article is a model of dependent type theory interpreting type classifiers.
Standard models are typically democratic, meaning that every context is represented by a closed type.
A context classifying all types would then yield a universe classifying all types, including itself, which is ruled out by Girard's paradox.
We use a simplified version of the big model construction for locally cartesian closed (lcc) categories introduced in~\citep{Bidlingmaier}.
Contexts in big models are themselves small models of type theory, in our case contextual categories with families (cwfs).
A type or term at a context $\C$ is a closed type or term of the cwf $\C$, and covariant substitutions are functors that preserve the cwf structure and type and term formers.

In the big model construction based on lcc categories, considerable effort is needed to obtain well-behaved context extensions by term variables.
Working with contextual cwfs simplifies this step: By a result of \citet*[Lemma 25]{BKS}, freely adjoining a term of a given closed type yields the contextual coslice, with a universal property for strict morphisms of cwfs.
Type classifiers are then obtained as coproducts with the free contextual cwf on one closed type.

As a category of categories, the big model also has 2-categorical structure: its substitutions are functors, and we can therefore consider natural transformations between substitutions.
We show that substituting isomorphic types into a type classifier gives naturally isomorphic substitutions, inducing isomorphisms between the resulting types.

This naturality can itself be expressed in the type theory by additional operators and equations.
The resulting extension gives an analogue of the invariance supplied by the univalence axiom in homotopy type theory, even for extensional type theory.
Omitting function types from the theory, so that all type formers are covariant, extends the naturality to arbitrary, possibly noninvertible maps.

\section{Preliminaries}

A \emph{covariant cwf} on a category $\C$ is an ordinary cwf structure~\citep[Section 2]{BCDE} on $\C^{\mathrm{op}}$, as in~\citep[Definition 14]{Bidlingmaier}.\footnote{
In the big model, substitution is given by an interpretation functor acting on types and terms.
We use covariant cwfs so that these functors themselves are the context morphisms instead of their formal opposites.
}
Thus $\C$ has a chosen initial \emph{empty context} $0_\C$, sets $\Ty_\C(\Gamma)$ of \emph{types}, and sets $\Tm_\C(\Gamma,\sigma)$ of \emph{terms}.
We write \emph{substitution} as function application: for $f:\Gamma\to\Delta$, the substituted type and term are $f(\sigma)\in\Ty_\C(\Delta)$ and $f(a)\in\Tm_\C(\Delta,f(\sigma))$.
Functoriality of substitution means that $g(f(\sigma))=(gf)(\sigma)$ and $\id_\Gamma(\sigma)=\sigma$, and similarly for terms.

We call ordinary context extensions \emph{term classifiers} to distinguish them from extensions by type parameters, as provided by the type classifiers introduced in Section~\ref{sec:type-classifiers}.
For $\sigma\in\Ty_\C(\Gamma)$, a term classifier consists of a context $\Gamma.\sigma$, a morphism and a generic term
\[
 p:\Gamma\longrightarrow\Gamma.\sigma,\qquad
 v\in\Tm_\C(\Gamma.\sigma,p(\sigma)).
\]
Its universal property is that, for $f:\Gamma\to\Delta$ and $a\in\Tm_\C(\Delta,f(\sigma))$, there is a unique morphism $\langle f,a\rangle:\Gamma.\sigma\to\Delta$ such that $\langle f,a\rangle p=f$ and $\langle f,a\rangle(v)=a$.

Let $\mathbb T$ be a dependent type theory presented as an essentially algebraic extension of cwfs by substitution-stable type and term formers.
Examples include theories with intensional or extensional identity types, dependent sums and products, unit types, and natural numbers.
Write $\Mod$ for the category of small covariant cwfs with this structure and \emph{strict morphisms}: Functors equipped with compatible actions on types and terms, preserving the chosen empty context, term classifiers, and type and term formers up to equality.

A \emph{telescope} over $\Gamma$ is a finite list of successive types, with evaluation $\Gamma.T$ and composite classifier map $p_T:\Gamma\to\Gamma.T$.
A cwf is \emph{contextual} if each context has a unique telescope presentation over $0_\C$.
Write $\Cxl$ for the full subcategory of contextual models.

We recall the contextual slice construction of \citet*[Appendix A.2]{BKS}, reversing all context arrows to follow our convention.
For any cwf $\C$ and context $\Gamma$ of $\C$, the \emph{contextual coslice} $\cs \C\Gamma$ is the contextual cwf whose contexts are telescopes over $\Gamma$ and whose morphisms $T\to S$ are maps $h:\Gamma.T\to\Gamma.S$ satisfying $h p_T=p_S$.
Types and terms at $T$, together with their formers, are those of $\C$ at $\Gamma.T$.
The empty telescope is initial, and a term classifier is obtained by appending the given type to the telescope.

For $f:\Gamma\to\Delta$, \emph{transport} $f_*:\cs \C\Gamma\to\cs \C\Delta$ is defined by iterating the pushout squares determined by term classifiers along a telescope.
The \emph{comparison map} at $T$ is denoted by $f_T:\Gamma.T\to\Delta.f_*(T)$.
The image of a telescope morphism is the induced map between these pushouts, while transport of types and terms is given by the action of $f_T$ in $\C$.
A strict morphism $F:\C\to \D$ induces a strict morphism $F_\Gamma:\cs \C\Gamma\to\cs \D{F(\Gamma)}$.
The assignment $f\mapsto f_*$ is functorial~\citep[after Definition 22]{BKS}:
\[
 (gf)_*=g_*f_*,\qquad (\id)_*=\mathrm{Id}.
\]
Transport is compatible with strict morphisms:
\[
 \begin{CD}
 \cs \C\Gamma @>{f_*}>> \cs \C\Delta \\
 @V{F_\Gamma}VV @VV{F_\Delta}V \\
 \cs \D{F(\Gamma)} @>{(F(f))_*}>> \cs \D{F(\Delta)}
 \end{CD}
\]
\emph{Evaluation} $e_\C:\cs \C{0_\C}\to \C$ is strict, and every strict morphism from a contextual model into $\C$ factors uniquely through it.
If $\C$ is contextual, $e_\C$ is a strict isomorphism~\citep[Section 2.3 and Corollary 24]{BKS}.
The coslice of $\cs \C\Gamma$ at $T$ is strictly isomorphic to $\cs \C{\Gamma.T}$ by concatenation of telescopes~\citep[proof of Lemma 26]{BKS}.

For a contextual $\C$ and a closed type $\sigma$, define
\[
 \C.\sigma=\cs \C{(0_\C.\sigma)}.
\]
The map $p:\C\to\C.\sigma$ is given by transport along the term-classifier map $0_\C\to0_\C.\sigma$, composed with $e_\C^{-1}$.
The generic term $v$ is a closed term of type $p(\sigma)$ in $\C.\sigma$.
The following is the universal property proved in~\citep[Lemma 25]{BKS}, expressed for covariant cwfs.

\begin{proposition}\label{prop:term}
For every $\D\in\Mod$ there is a natural bijection
\begin{equation}\label{eq:term}
 \Mod(\C.\sigma,\D)\cong
 \coprod_{F:\,\C\to \D}\Tm_\D(0_\D,F(\sigma)),
 \qquad H\longmapsto(Hp,H(v)).
\end{equation}
\end{proposition}
\section{Type classifiers and the big model}\label{sec:type-classifiers}

Our goal is to give an interpretation of a type theory in which contexts can be extended by free type variables.
The precise presentation depends on the other rules of the type theory, but one possible presentation is by the following rule set:
\[
 \begin{gathered}
 \frac{\vdash\Gamma\ \mathsf{ctx}}
      {\vdash\Gamma,X\ \mathsf{type}\ \mathsf{ctx}}
 \qquad
 \frac{\vdash\Gamma\ \mathsf{ctx}}
      {\Gamma,X\ \mathsf{type}\vdash X\ \mathsf{type}}
 \\[1.5ex]
 \frac{\Gamma\vdash\mathcal J}
      {\Gamma,X\ \mathsf{type}\vdash\mathcal J}
 \qquad
 \frac{\Gamma,X\ \mathsf{type},\Delta(X)\vdash\mathcal J(X)
       \quad\Gamma\vdash\sigma\ \mathsf{type}}
      {\Gamma,\Delta(\sigma)\vdash\mathcal J(\sigma)}
 \end{gathered}
 .
\]
Here $\mathcal J$ ranges over type and term judgments and their judgmental equalities.
The expressions $\Delta(\sigma)$ and $\mathcal J(\sigma)$ denote replacement of $X$ by $\sigma$, avoiding variable capture.

Regardless of the precise syntactic presentation, we expect such context extensions to be interpreted in covariant cwfs as type classifiers:
\begin{definition}
A \emph{type classifier} at a context $\Gamma$ of a covariant cwf consists of a context $\Gamma[X]$, a morphism $p_\Gamma:\Gamma\to\Gamma[X]$, and a type $X\in\Ty(\Gamma[X])$, with the following universal property: for every $f:\Gamma\to\Delta$ and $\sigma\in\Ty(\Delta)$ there is a unique morphism $\langle f,\sigma\rangle:\Gamma[X]\to\Delta$ such that
\begin{equation}\label{eq:type-laws}
 \langle f,\sigma\rangle p_\Gamma=f,\qquad \langle f,\sigma\rangle(X)=\sigma.
\end{equation}
\end{definition}
The following is our main result.

\begin{theorem}\label{thm:big}
The category $\Cxl$ has a covariant cwf structure with
\[
 \Ty_{\Cxl}(\C)=\Ty_\C(0_\C),\qquad
 \Tm_{\Cxl}(\C,\sigma)=\Tm_\C(0_\C,\sigma).
\]
It supports the type and term formers of $\mathbb T$ and type classifiers at all contexts.
\end{theorem}
\begin{proof}
The empty context is the initial contextual model.
For the term classifier at $\sigma$, take $\C.\sigma$ with the structure map $p$ and generic term $v$ of Proposition~\ref{prop:term}.

By definition of the contextual coslice, we have the literal equalities
\begin{equation}\label{eq:closed}
 \Ty_{\C.\sigma}(0)=\Ty_\C(0_\C.\sigma),\qquad
 \Tm_{\C.\sigma}(0,\tau)=\Tm_\C(0_\C.\sigma,\tau).
\end{equation}
By iterating~\eqref{eq:closed}, types and terms over successive term classifiers in the big model can be identified with types and terms in the corresponding telescopes of $\C$.
The type and term formers can thus be defined using those of $\C$, with the same equations.
Compatibility with substitution in the big cwf follows from compatibility of transport with strict morphisms and their preservation of the type and term formers.

For type classifiers, let $U$ be the free contextual $\mathbb T$-cwf on a single closed type.
Such a free model and coproducts exist because contextual $\mathbb T$-cwfs admit an essentially algebraic presentation~\citep[Section 5.7]{BCDE}.
Put $\C[X]=\C\amalg U$, with $p_\C$ the coproduct inclusion and $X$ the image of the generating type of $U$.
By the coproduct universal property, this is a type classifier at $\C$.
\end{proof}

Now consider isomorphic types $\sigma\cong\tau$ in a context $\Gamma$ of the big model and the induced functors $\bar\sigma,\bar\tau:\Gamma[X]\to\Gamma$, where $\bar\sigma=\langle\mathrm{Id}_\Gamma,\sigma\rangle$ and $\bar\tau=\langle\mathrm{Id}_\Gamma,\tau\rangle$.
Does the isomorphism $\bar\sigma(X)\cong\bar\tau(X)$ extend to isomorphisms $\bar\sigma(\kappa)\cong\bar\tau(\kappa)$ for every type $\Gamma[X]\vdash\kappa\ \mathsf{type}$?
If the type formers of $\mathbb T$ are well-behaved, the answer is yes, by the following theorem.

\begin{theorem}\label{thm:naturality}
Suppose the type and term formers of $\mathbb T$ are functorial with respect to invertible maps between types.
Then for every strict $F:\C\to\D$ between contextual models and closed types $\sigma,\tau$ of $\D$, there is a natural bijection
\[
 \operatorname{Iso}_{F}\bigl(\langle F,\sigma\rangle,\langle F,\tau\rangle\bigr)
 \cong \operatorname{Iso}(\sigma,\tau).
\]
If the type and term formers are functorial with respect to arbitrary maps between types, there is also a natural bijection
\[
 \operatorname{Hom}_{F}\bigl(\langle F,\tau\rangle,\langle F,\sigma\rangle\bigr)
 \cong \operatorname{Hom}(\sigma,\tau).
\]
\end{theorem}

On the left, $\operatorname{Iso}_{F}$ denotes $\mathbb T$-natural isomorphisms between the indicated functors whose restriction along $p_\C$ is the identity on $F$.
Here $\mathbb T$-naturality means that the comparison for each type former is the one induced by the comparisons of its arguments.
For type formers characterized by universal diagrams of context morphisms, this condition follows from naturality.
On the right, $\operatorname{Iso}(\sigma,\tau)$ denotes invertible maps between types in $\D$, identified with isomorphisms $0_\D.\sigma\cong0_\D.\tau$ of their term classifiers.
Likewise, $\operatorname{Hom}_{F}$ denotes $\mathbb T$-natural transformations restricting to the identity on $F$, and $\operatorname{Hom}(\sigma,\tau)$ denotes context morphisms $0_\D.\tau\to0_\D.\sigma$, or equivalently open terms $x:\sigma\vdash s(x):\tau$.
The reversed order in the second bijection is thus due to our use of covariant rather than contravariant cwfs.

We leave the notion of functoriality of type and term formers informal so as to avoid having to define a general notion of dependent type theory or type former.
Informally, given compatible (invertible) maps between the arguments of a type former, we expect an (invertible) map between the resulting types, and the corresponding term formers should respect these maps.
These actions are assumed to respect the equations of $\mathbb T$ and to commute with substitution and strict morphisms.
In extensional type theory, unit types, finite sums, dependent sums, identity types, and natural numbers are functorial even with respect to noninvertible maps.
Dependent products require invertible maps owing to the mixed variance of their arguments.

Intensional identity types require care because the functoriality laws must hold judgmentally.
For the usual action $\mathsf{ap}$ defined by identity elimination, even $\mathsf{ap}_{\id}(p)\equiv p$ need not hold judgmentally for a variable $p$.
Strictly inverse maps between the underlying types therefore do not by themselves give judgmentally inverse maps on identity types by this construction.
For cubical path types, by contrast, the induced maps typically satisfy the identity and composition laws judgmentally; see, for example,~\citep[Sections 3.1--3.2]{CCHM}.
A similar issue occurs for intensional $W$-types, whose usual computation rules do not imply judgmental uniqueness of recursion.

\begin{proof}[Proof of Theorem~\ref{thm:naturality}]
With $\mathbb T$-natural transformations as $2$-cells, $\Cxl$ is a $2$-category in the arbitrary-map case, and with $\mathbb T$-natural isomorphisms it is a $(2,1)$-category in the invertible case.
The assertion of the theorem is that the coproduct defining a type classifier also has the corresponding $2$-categorical universal property.
For this, it suffices to construct cotensors by the walking (invertible) arrow, as we will see below.
Write $\D^{\rightarrow}$ and $\D^{\cong}$ for these cotensors.
Their universal property identifies strict morphisms $\C\to\D^{\rightarrow}$ with $\mathbb T$-natural transformations of strict morphisms $\C\to\D$, and similarly for $\D^{\cong}$ and $\mathbb T$-natural isomorphisms.

Contexts of $\D^{\cong}$ are pairs of telescopes of the same length, equipped with isomorphisms between corresponding prefixes that commute with the term-classifier maps.
\[
 \begin{CD}
  0_\D @>{p}>> 0_\D.\sigma_1 @>>> \cdots @>{p}>> 0_\D.\sigma_1.\cdots.\sigma_n\\
  @| @V{u_1}V{\cong}V @. @V{u_n}V{\cong}V\\
  0_\D @>{p}>> 0_\D.\tau_1 @>>> \cdots @>{p}>> 0_\D.\tau_1.\cdots.\tau_n
 \end{CD}
\]

Morphisms are commuting squares between the isomorphisms of the full contexts.
A type consists of a type over each endpoint together with an isomorphism of their term classifiers over the given context isomorphism, and a term consists of a compatible pair of terms.
Context extension is defined by appending these data to the two telescopes.
The type and term formers are defined at the endpoints, with the comparison isomorphisms supplied by their functoriality.
For $\D^{\rightarrow}$, the same construction uses maps in place of isomorphisms.

A map $u:\sigma\to\tau$ of closed types in $\D$ defines a closed type in $\D^{\rightarrow}$.
This type and the composite $\C\xrightarrow{F}\D\to\D^{\rightarrow}$ given by the diagonal determine a strict morphism $\C[X]\to\D^{\rightarrow}$, hence a $\mathbb T$-natural transformation from $\langle F,\tau\rangle$ to $\langle F,\sigma\rangle$.
Conversely, a map $\sigma\to\tau$ can be obtained by evaluation at the generic type $X$.
For invertible maps, use $\D^{\cong}$ in place of $\D^{\rightarrow}$.
\end{proof}

Naturality of type classifiers can be internalized with an operator $\ap_{X,x}$ corresponding to the components of $\mathbb T$-natural transformations between type-substitution functors.
Its introduction rule is given by
\begin{equation}\label{eq:ap-rule}
 \frac{\begin{gathered}
 \Gamma\vdash\sigma\ \mathsf{type}\qquad
 \Gamma\vdash\tau\ \mathsf{type}\\
 \Gamma,X\ \mathsf{type}\vdash P(X)\ \mathsf{type}\qquad
 \Gamma,x:\sigma\vdash s(x):\tau
 \end{gathered}}
 {\Gamma,z:P(\sigma)\vdash\ap_{X,x}(P,s,z):P(\tau)}
\end{equation}
in case the type formers of the theory are functorial with respect to all maps, with additional hypotheses corresponding to the existence of a definitional inverse in the invertible case:
\[
 \Gamma,y:\tau\vdash t(y):\sigma\qquad
 \Gamma,x:\sigma\vdash t(s(x))\equiv x:\sigma\qquad
 \Gamma,y:\tau\vdash s(t(y))\equiv y:\tau
\]
Here $X$ is the type variable in $P(X)$, and $x$ is the term variable in $s(x)$.
In $\ap_{X,x}(P,s,z)$, $X$ is bound in $P$ and $x$ is bound in $s$.
Parenthesized arguments indicate substitution.
The intended interpretations of terms $\ap_{X,x}(P,s,z)$ in the big model are the components of $\mathbb T$-natural transformations $\langle\mathrm{Id},\tau\rangle\Rightarrow\langle\mathrm{Id},\sigma\rangle:\C[X]\to\C$ between type-substitution functors given by Theorem~\ref{thm:naturality}.

Evaluation at the generic type and uniqueness in Theorem~\ref{thm:naturality} give identity and composition laws for type families.
These laws are
\[
 \frac{\begin{gathered}
 \Gamma\vdash\sigma\ \mathsf{type}\qquad
 \Gamma\vdash\tau\ \mathsf{type}\\
 \Gamma,X\ \mathsf{type}\vdash P(X)\equiv X\ \mathsf{type}\qquad
 \Gamma,x:\sigma\vdash s(x):\tau
 \end{gathered}}
 {\Gamma,z:\sigma\vdash\ap_{X,x}(P,s,z)\equiv s(z):\tau}
\]
for identity and
\[
 \frac{\begin{gathered}
 \Gamma\vdash\sigma\ \mathsf{type}\qquad
 \Gamma\vdash\tau\ \mathsf{type}\\
 \Gamma,X\ \mathsf{type}\vdash P(X)\ \mathsf{type}\qquad
 \Gamma,Y\ \mathsf{type}\vdash Q(Y)\ \mathsf{type}\qquad
 \Gamma,x:\sigma\vdash s(x):\tau
 \end{gathered}}
 {\Gamma,z:Q(P(\sigma))\vdash
 \ap_{X,x}(Q(P(X)),s,z)\equiv\ap_{Y,u}(Q,\ap_{X,x}(P,s,u),z):Q(P(\tau))}
\]
for composition.
The operator $\ap_{X,x}$ is expected to commute with identity and composition in $s$, corresponding to preservation of identities and composition of $2$-cells in Theorem~\ref{thm:naturality}:
\[
 \begin{gathered}
 \frac{\Gamma,X\ \mathsf{type}\vdash P(X)\ \mathsf{type}
 \qquad \Gamma\vdash\sigma\ \mathsf{type}
 \qquad \Gamma,x:\sigma\vdash s(x)\equiv x:\sigma}
 {\Gamma,z:P(\sigma)\vdash\ap_{X,x}(P,s,z)\equiv z:P(\sigma)}
 \\[1.2em]
 \frac{\begin{gathered}
 \Gamma\vdash\sigma\ \mathsf{type}\qquad
 \Gamma\vdash\tau\ \mathsf{type}\qquad
 \Gamma\vdash\kappa\ \mathsf{type}\\
 \Gamma,X\ \mathsf{type}\vdash P(X)\ \mathsf{type}\\
 \Gamma,x:\sigma\vdash s(x):\tau\qquad
 \Gamma,y:\tau\vdash t(y):\kappa
 \end{gathered}}
 {\Gamma,z:P(\sigma)\vdash\ap_{X,x}\bigl(P,t(s(x)),z\bigr)
 \equiv\ap_{X,y}\bigl(P,t,\ap_{X,x}(P,s,z)\bigr):P(\kappa)}
 \end{gathered}
\]
Naturality with respect to terms is expressed by the following rule:
\[
 \frac{\begin{gathered}
 \Gamma\vdash\sigma\ \mathsf{type}\qquad
 \Gamma\vdash\tau\ \mathsf{type}\\
 \Gamma,X\ \mathsf{type}\vdash P(X)\ \mathsf{type}\\
 \Gamma,X\ \mathsf{type}\vdash r(X):P(X)\qquad
 \Gamma,x:\sigma\vdash s(x):\tau
 \end{gathered}}
 {\Gamma\vdash\ap_{X,x}(P,s,r(\sigma))\equiv r(\tau):P(\tau)}
\]
As usual, the operator is required to respect $\alpha$-conversion of the bound variables $X$ and $x$ and to be congruent with respect to definitional equality in its arguments.

\bibliographystyle{plainnat}
\bibliography{type-classifiers}
\end{document}